\documentclass[a4paper,12pt]{amsart}

\numberwithin{equation}{section}

\usepackage{amsmath}

\usepackage{amsfonts,amssymb,amsmath}
\usepackage{enumerate, xspace}
\usepackage{changebar}
\usepackage[colorlinks]{hyperref}

\newtheorem{theorem}{Theorem}[section]
\newtheorem{lemma}[theorem]{Lemma}

\theoremstyle{definition}

\newcommand{\mc}[1]{{\mathcal #1}}
\newcommand{\bb}[1]{{\mathbb #1}}
\newcommand{\mb}[1]{{\mathbf #1}}
\newcommand{\R}{\mathbb R} 
\newcommand{\Z}{\mathbb Z}
\newcommand{\gl}{g_\lambda}
\newcommand{\Gl}{G_\lambda}
\newcommand{\ve}{\varepsilon}

\newcommand{\bx}{\mb x}
\newcommand{\by}{\mb y}
\newcommand{\be}{\mb e}
\newcommand{\bp}{\mb p}
\newcommand{\bq}{\mb q}
\newcommand{\bz}{\mb z}

\begin{document}

\title{Energy Diffusion in Harmonic System with Conservative Noise}
\author{Giada Basile}
\address{Giada Basile\\
 Dipartimento di Matematica\\
 Universit\`{a} di Roma La Sapienza\\
 Roma, Italy.}
 \email{{\tt basile@mat.uniroma.it}
}
\author{Stefano Olla}
\address{Stefano Olla\\
 CEREMADE, UMR-CNRS 7534\\
 Universit\'{e} Paris Dauphine\\
 Paris, France.}
 \email{{\tt olla@ceremade.dauphine.fr}
}

\date{\today.}
\begin{abstract}
We prove diffusive behaviour of the energy fluctuations in a system of
harmonic oscillators with a stochastic perturbation of the dynamics
that conserves energy and momentum. The results concern pinned systems
or lattice dimension $d\ge 3$, where the thermal diffusivity is finite.
\end{abstract}
\thanks{Dedicated to Herbert Spohn for his 65th birthday.\\
This paper has been partially supported by the
  European Advanced Grant {\em Macroscopic Laws and Dynamical Systems}
  (MALADY) (ERC AdG 246953)}
\subjclass[2000]{...}
\maketitle

\section{Introduction}
\label{sec:introduction}

Lattice networks of oscillators have been considered for a long time as
good models for studying macroscopic energy trasport and its
diffusion, i.e. for obtaining, on a macroscopic space-time scale, heat
equation and Fourier law of conduction (\cite{blr}).   
It is well understood that the diffusive behavior of the energy is due
to the non-linearity of the interactions, and that purely
deterministic harmonic systems have a ballistic transport of energy
(cf. \cite{rll}). 
On the other hand, non-linear dynamics are very difficult to study
and even the 
convergence of the Green-Kubo formula defining the macroscopic thermal
conductivity is an open problem. 
In fact in some cases, like in one dimensional un-pinned
systems, it is expected (and observed numerically) an infinite
conductivity and a superdiffusion of the energy (cf. \cite{llp97,
  sll}). 

In order to model the phonon scattering due to the dynamics,
 various stochastic perturbation of the harmonic
dynamics have been proposed, 
where the added random dynamics conserves the energy.  
In \cite{bo} Fourier Law is proven for an harmonic chain with 
 stocastic dynamics that conserves only energy. In \cite{bborev, bbo2}
 is studied the Green-Kubo formula for a stochastic perturbation that
 conserves energy \emph{and} momentum. It is proven there that
 conductivity is finite for pinned systems, or unpinned if dimension
 is greater than 2. Qualitatively this agrees to what it is expected
 for deterministic anharmonic dynamics. 

In this article we consider the same stochastic dynamics as in
\cite{bborev, bbo2} and we prove that in the cases when conductivity
is finite, the energy behave diffusively in the sense that energy
fluctuations of the system in equilibrium evolve according to a linear
heat equation. 

The key point in proving such diffusive behavior is to obtain a
fluctuation-dissipation decomposition of the microscopic energy
currents $j_{\bx, \by}$ between two adiacent atoms $\bx, \by$. This
means to be able to approximate $j_{\bx, \by}$ by a function of the
form $\kappa (e_\bx - e_\by) + L F$, where $e_\bx$ is the energy of
atom $\bx$, $L$ is the generator of the dynamics and $F$ a function in
its domain, possibly local. With such decomposition it is possible to
close macroscopically the evolution of the energy and $\kappa$ identify
the thermal diffusivity.

In the harmonic chain with noise that conserves only energy, there is
an exact fluctuation-dissipation decomposition with a local function
$F$ (cf. \cite{bo}). For anharmonic system this decomposition is
non-local, and much harder to obtain. In fact there exist only
results with an \emph{elliptic} noise that acts also on the positions
of the atoms (cf. \cite{os}). 

In the model considered here, the noise is of hypoelliptic type,
i.e. acts only on the velocity. But because of the additional
conservation of momentum, the fluctuation-dissipation decomposition of
the currents is non-local. Thanks to the linearity of the interaction,
we can use resolvent approximations instead of local approximations
for this decomposition.  

The corresponding results for the cases when thermal conductivity is
infinite, i.e. the unpinned model in dimension 1 and 2, remains open
problems. In dimension 2 the divergence of the conductivity is of
\emph{logarithmic} type, and we expect a diffusive behavior of the
energy fluctuations under the proper space-time scale. 

In dimension 1, also to establish a conjecture is hard, as the
behavior is really superdiffusive. Under a weak noise limit the local
spectral energy density (Wigner distribution)
behaves following a linear Boltzmann type equation (cf. \cite{bos}). 
 Under proper rescaling of this equation, is proven the convergence to
 a superdiffusive Levy process governed by a fractional laplacian
 (cf. \cite{babo, kjo}. See also \cite{ba} for diffusive behaviour in dimension 2). A
 possible guess is that the fractional laplacian behavior will govern
 energy fluctuations without taking first a weak noise limit. 
For what concerns this problem in one dimensional unpinned anharmonic
chains, see the recent article by Herbert Spohn \cite{Sp13}.

\section{The dynamics}
\label{sec:dynamics}

The Hamiltonian is given by
\begin{equation}
  \label{eq:hamilt}
  \mathcal H = \frac 12 \sum_{\bx}  \left[ {\bp_\bx^2} +
\bq_\bx \cdot (\nu I - \alpha\Delta) \bq_\bx \right] . \nonumber
\end{equation}
The atoms are labeled by $\bx \in \mathbb Z^d$ and $\{\bq_\bx\}$ are the displacements of the atoms 
from their equilibrium positions.
 We denote with $\nabla$, $\nabla^*$ and 
$\Delta = \nabla^*\cdot \nabla$
respectively the discrete gradient, its adjoint and the discrete Laplacian  on
$\mathbb Z^d$. These are defined as
\begin{equation}
  \label{eq:grad}
  \nabla_{\be_j} f(\bx) = f(\bx + \be_j) - f(\bx)
\end{equation}
and
\begin{equation}
  \label{eq:adjgrad}
   \nabla^*_{\be_j} f(\bx) =f(\bx - \be_j)- f(\bx).
\end{equation}

 The parameter $\alpha > 0$ is the strength of the
interparticles springs, and $\nu \ge 0$ is the strength of the pinning
(on-site potential).

We consider the stochastic dynamics corresponding to the Fokker-Planck
equation
\begin{equation}\label{FP}
 \frac {\partial P}{\partial t} = (- A + \gamma S) P =  L P\ .
\end{equation}
where $A$ is the usual Hamiltonian vector field
 \begin{equation*}
  \label{eq:Agen}
  \begin{split}
    A =& \sum_\bx \left\{ \bp_\bx \cdot \partial_{\bq_\bx} +
     [ (\alpha\Delta - \nu I)
    \bq_\bx] \cdot  \partial_{\bp_\bx} \right\}
\end{split}
\end{equation*}
while $S$ is the generator of the stochastic perturbation and 
$\gamma > 0$ is a positive parameter that regulates its strength.
The operator $S$ acts only on the momentums $\{\bp_\bx\}$ and
generates a diffusion on the surface of constant kinetic energy and
constant momentum. This is defined as follows. For every nearest
neighbor atoms $\bx$ and $\bz$, consider the $d-1$ dimensional surface 
of constant kinetic energy and momentum
\begin{equation*}
  \mathbb S_{e,\bp} \; = \; \left\{(\bp_\bx,\bp_\bz)\in \mathbb R^{2d}:
    \frac 12\left(\bp_\bx^2 + \bp_\bz^2\right) = e\; ; \; 
    \bp_\bx + \bp_\bz = \bp  \right\}\ .
\end{equation*}
The following vector fields are tangent to $ \mathbb S_{e,\bp}$
\begin{equation*}
  \label{eq:Xfield}
   X^{i,j}_{\bx, \bz} = (p^j_\bz-p^j_\bx) (\partial_{p^i_\bz}
   - \partial_{p^i_\bx})  -(p^i_\bz-p^i_\bx) (\partial_{p^j_\bz}
   - \partial_{p^j_\bx}) .
\end{equation*}
so $ \sum_{i,j =1}^d (X^{i,j}_{\bx, \bz})^2$ generates a diffusion on
$ \mathbb S_{e,\bp}$. 
In  $d\ge 2$ we define
\begin{equation*}
  \label{eq:Sgen}\begin{array}{ll}
  S & = \displaystyle\frac 1{2(d-1)} \sum_{\bx}\sum_{i,j,k}^d
  \left( X^{i,j}_{\bx, \bx+\be_k}\right)^2 \\
  & = \displaystyle\frac  1{4(d-1)} \sum_{\bx, \bz \in \mathbb Z^d_N \atop \|\bx -
 \bz\|=1}\sum_{i,j}
   \left( X^{i,j}_{\bx, \bz}\right)^2
\end{array}
\end{equation*}
where  ${\be}_1,\ldots,{\be}_d$ is canonical basis of ${\mathbb Z}^d$.

Observe that this noise conserves the total momentum $\sum_\bx
\bp_\bx$ and energy $\mathcal H_N$, i.e.  
\begin{equation*}
  S \; \sum_\bx \bp_\bx = 0\ ,\quad S\; \mathcal H_N = 0
\end{equation*}

In dimension 1, in order to conserve total momentum and total kinetic energy, 
we have to consider a random exchange of momentum between three 
consecutive atoms, and we define
$$
S = \frac 16 \sum_{x\in\mathbb{T}^1_N}(Y_x)^2
$$
where
\begin{equation*}
\label{eq:005}
Y_x=(p_x-p_{x+1})\partial_{p_{x-1}}+(p_{x+1}-p_{x-1})\partial_{p_x} +
(p_{x-1}-p_x)\partial_{p_{x+1}} 
\end{equation*}
which is vector field tangent to the surface of constant energy and
momentum of the three particles involved.
The Fokker-Planck equation (\ref{FP}) gives the time evolution of the
probability distribution $P(\bq, \bp, t)$, given an initial
distribution  $P(\bq, \bp, 0)$. It correspond to the law at time t of the
solution of the following stochastic differential equations:
\begin{equation}
  \label{eq:sde}
  \begin{split}
     d\bq_\bx &= \bp_\bx\; dt\\
     d\bp_\bx &= -(\nu I - \alpha \Delta)\bq_\bx\; dt +
     2 \gamma  \Delta \bp_\bx \; dt \\
     & \qquad \qquad \qquad + \frac {\sqrt{\gamma}}{2\sqrt{ d-1}}
      \sum_{\bz: \|\bz - \bx\|=1} \sum_{i,j=1}^d \left(X^{i,j}_{\bx, \bz}
     \bp_{\bx}  \right) \; dw^{i,j}_{\bx, \bz}(t) 
  \end{split}
\end{equation}
where $\{w^{i,j}_{\bx, \by} = w^{i,j}_{\by, \bx};\; \bx, \by \in \mathbb
Z^d;\; i,j= 1,\dots, d;\; \|\by - \bx\| = 1\}$ are independent standard
Wiener processes.
In $d=1$ the sde are:
\begin{equation}
  \label{eq:sde1}
  \begin{split}
    dp_x = -(\nu I - \alpha \Delta)q_x\; dt + \frac \gamma 6
    \Delta(4p_x + p_{x-1} + p_{x+1}) dt \\
    + \sqrt{\frac \gamma 3}
    \sum_{k=-1,0, 1} \left(Y_{x+k} p_{x} \right) dw_{x+k}(t)
  \end{split}
\end{equation}
where here $\{w_{x}(t), x= 1, \dots, N\}$ are independent standard
Wiener processes.

Defining the energy of the atom $\bx$ as
\begin{equation*}
  \label{eq:energyx}
  e_\bx = \frac 12  \bp_\bx^2 +
\cfrac{\alpha}{4}\sum_{\by: |\by -\bx|=1}
(\bq_{\by} - \bq_\bx)^2 + \frac{\nu}2 \bq_\bx^2\ ,
\end{equation*}
the energy conservation law can be read locally as
\begin{equation*}
   e_\bx(t) -   e_\bx(0) =  \sum_{k=1}^d \left(
J_{\bx-\be_k, \bx}(t) - J_{\bx,\bx +\be_k}(t)\right)
 \end{equation*}
where $J_{\bx,\bx +\be_k}(t)$ is the total energy current
between $\bx$ and $\bx +\be_k$ up to
time $t$. This can be written as
\begin{equation}
  \label{eq:tc}
  J_{\bx, \bx +\be_k}(t) = \int_0^t j_{\bx, \bx +\be_k}(s) \; ds +
  M_{\bx, \bx +\be_k}(t) \ .
\end{equation}
In the above $M_{\bx, \bx +\be_k}(t)$ are  martingales that can be
written explicitly as Ito stochastic integrals
\begin{equation}
  \label{eq:mart}
  M_{\bx, \bx +\be_k}(t) = \sqrt{\frac{\gamma}{(d-1)}} \sum_{i,j}
  \int_0^t \left(X^{i,j}_{\bx, \bx +\be_k} e_\bx\right)(s) \;
  dw^{i,j}_{\bx, \bx +\be_k} (s) 
\end{equation}

The stationary equilibrium probability measures for this dynamics are
given by the corresponding Gibbs measure, defined through the DLR
equations. Because of the linearity of the interaction and the
conservation laws of the stochastic perturbation, these are gaussian
measures on $(\mathbb R^{2d})^{\mathbb Z^d}$ with covariance given by 
\begin{equation}
  \label{eq:8}
  \begin{split}
    <(p_{\bx}^i- v^i) (p_{\by}^j- v^j)> = \beta^{-1}
    \delta_{i,j}\delta_{\bx,\by}, \quad <(p_{\bx}^i- v^i) q_{\by}^j> =
    0,\\
    <q_{\bx}^i q_{\by}^j> =\frac 1d \Gamma(\bx - \by) \delta_{i,j}
  \end{split}
\end{equation}
where $\Gamma (\bx) = (\nu I - \alpha \Delta)^{-1} (\bx)$. In the
pinned case, $\nu =0$, momentum is not conserved  and we have to set
$\mb v=0$. Since in the unpinned case the parameter $\mb v$ represent
a trivial translation invariance, for simplicity we choose $\mb v = 0$.

We consider this dynamics starting with an equilibrium distribution at
a given temperature $\beta^{-1}$. The existence of the infinite
dynamics under this initial distribution can be proven by standard
techniques (for example see \cite{lll}).

Given two continuous functions $F, H$ on $\R^d$ with compact support,
and  $\ve>0$, we look at the evolution of the 
following quantity 
\begin{equation*}\begin{split}
\sigma_{\ve,t}\left(F,H\right) = &\ve^{d}\sum_{\mb x, \mb y}F(\ve
\mb x)H(\ve \mb y)\langle 
\left(e_{\mb x}(\ve^{-2}t) - \beta^{-1}\right)
\big(e_{\mb y}(0)-\beta^{-1}\big)\rangle\\
&=\ve^{d} \sum_{\mb y,\mb z}F\big(\ve (\mb y+\mb z)\big)H(\ve
\mb y)\langle e_{\mb z}(\ve^{-2}t) 
\big(e_{\mb 0}(0)-\beta^{-1}\big)\rangle,
\end{split}\end{equation*}
where $\langle\;\cdot\;\rangle$ is the expectation value with respect
to the equilibrium measure.

\begin{theorem}\label{mainth}
  \begin{equation}
  \lim_{\ve \to 0}  \sigma_{\ve,t}\left(F,H\right) = 
\iint d\mb u\; d\mb v\ F(\mb u)\; G(\mb v) \;
\frac{e^{-|\mb u - \mb v|^2/2t\kappa}}{\left( 2\pi t \kappa \right)^{d/2}}\label{eq:9}
\end{equation}
where the diffusion coefficient $\kappa$ is given by:
\begin{equation}
\kappa=\frac 1 {8\pi^2\gamma}\int_{\bb T^d}d\xi \,
\frac{\left(\partial_1\omega(\xi) \right)^2}{\Phi(\xi)} + \gamma
\end{equation}
Here $\omega$ is the dispersion relation
$$
\omega(\xi)=\big(\nu+4\alpha\sum_{j=1}^d\sin^2(\pi\xi_j)
\big)^{1/2}
$$
and $\Phi$ is the scattering rate
\begin{equation}\label{def:phi}
\Phi(\mb k)=\left\{\begin{array}{ll} \vspace{0.2cm}
8\sum_{j=1}^d \sin^2(\pi k_j) & d\geq 2\\
\frac 4 3 \sin^2(\pi k)\big(1+2\cos^2(\pi k)\big) &d=1
\end{array}\right.
\end{equation}
\end{theorem}

With a little more work one can prove that the fluctuation field
\begin{equation}
  \label{eq:10}
  Y^\ve_t(F) = \ve^{d/2} \sum_{\mb y}F\big(\ve \mb y\big) 
  \left[e_{\mb z}(\ve^{-2}t) -\beta^{-1}\right]
\end{equation}
converges in law to the infinite dimensional Ornstein Uhlenbeck $Y_t$
solution of the linear stochastic PDE:
\begin{equation}
  \label{eq:15}
  \partial_t Y = \frac{\kappa}2 \Delta Y + \beta^{-1/2} \nabla W
\end{equation}
where $W(\bx, t)$ is the standard space time white noise on
$\R^{d+1}$. This extension is standard and we will expose here only
the proof of Theorem  \ref{mainth}.

\section{Energy currents}
Let $e_{\bf x}$ be the energy of the atom $\bf x$, which is equal to
\begin{equation*}
e_{\bf x}=\frac 1 2 \mb p _{\bf x} ^2-\frac 1 2 \mb q_{\mb x}\cdot\big(\alpha\Delta -\nu\big)
\bf q_{\bf x}.
\end{equation*}
When  $\nu=0$ there is no pinning. We consider cases  $\nu>0$, $d\geq 1$ and $\nu=0$, $d=3$.
The  instantaneous energy currents $j_{\mb x, \mb x+\mb e_i}$, $i=1,..,d$, satisfy the equation
\begin{equation*}
L e_{\bf x}=\sum_{i=1}^d\big(j_{\mb x-\mb e_i, \bf x }-j_{\mb x, \mb x+\mb e_i } \big),
\end{equation*}
and  can be written as
\begin{equation}\label{def:j}\begin{split}
j_{\mb x, \mb x+\mb e_i }=j^a_{\mb x, \mb x+\mb e_i }+ \gamma j^s_{\mb x, \mb x+\mb e_i }.
\end{split}\end{equation}
The first term is the Hamiltonian contribution to the energy current, namely
\begin{equation}\label{def:ja}
j^a_{\mb x, \mb x+\mb e_i }=-\frac{\alpha} 2 
\big(\mb q_{\mb x+\mb e_i}-\mb q_{\mb x}\big)
\cdot\big(\mb p_{\mb x+\mb e_i}+\mb p_{\mb x}\big),
\end{equation}
while the noise contribution in $d\geq 2$ is
\begin{equation}\label{def:js}
\gamma j^s_{\mb x, \mb x+\mb e_i }=-\gamma \nabla_{\mb e _i}\mb p^2_{\mb x}.
\end{equation}
In one dimension
\begin{equation}\label{def:js1}\begin{split}
\gamma j^s_{x,x+1}= & -\gamma \nabla\phi(p_{x-1},p_x p_{x+1}),\\
\phi(p_{x-1},p_x p_{x+1})= & \frac 1 6 \big[p_{x+1}^2 +4p_x^2
    +p_{x-1}^2 +p_{x+1}p_{x-1}\\
&-2p_{x+1}p_{x}-2p_{x-1}p_x\big].
\end{split}\end{equation}
We  denote with $\phi_x:=\phi(p_{x-1},p_x,p_{x+1})$.

Given $F, H\in C_c^2(\R^d)$ (twice differentiable functions with
compact support), and  $\ve>0$, we look at the evolution of the
following quantity 
\begin{equation*}\begin{split}
\sigma_{\ve,t}\left(F,H\right) = &\ve^{d}\sum_{\mb x, \mb y}F(\ve
\mb x)H(\ve \mb y)\langle 
\left(e_{\mb x}(\ve^{-2}t) - \beta^{-1}\right)
\big(e_{\mb y}(0)-\beta^{-1}\big)\rangle\\
&=\ve^{d} \sum_{\mb y,\mb z}F\big(\ve (\mb y+\mb z)\big)H(\ve
\mb y)\langle e_{\mb z}(\ve^{-2}t) 
\big(e_{\mb 0}(0)-\beta^{-1}\big)\rangle,
\end{split}\end{equation*}
where $\langle\;\cdot\;\rangle$ is the expectation value with respect to the equilibrium measure.
For $d\geq 1$ we have
\begin{equation}\label{main}
\begin{split}
& \sigma_{\ve,t}\left(F,H\right) = \sigma_{\ve,0}\left(F,H\right)\\
&+\ve^d \sum_{\mb y, \mb z}\sum_{i=1}^d
\nabla^\ve_{\mb e_i}F\big(\ve (\mb y+\mb z)\big)H(\ve \mb y)\frac 1 {\ve}
\int_0^{t} ds\;\left< j_{\mb z,\mb z+\mb e_i}(s/\ve^2)\big(e_{\mb
    0}(0)-\beta^{-1}\big)\right>.
\end{split}\end{equation}
Accordingly to \eqref{def:j}, we decompose the energy current $j_{\mb z,\mb z+\mb e_i}$
in two parts and we treat separately the two integrals.
\subsection*{Noise current}
For $d\geq 2$, using \eqref{def:js} we have
\begin{equation*}\begin{split}
&\gamma \ve^d \sum_{\mb y, \mb z}\sum_{i=1}^d
\nabla^\ve_{\mb e_i}F\big(\ve (\mb y+\mb z)\big)H(\ve \mb y)\frac 1 {\ve}
\int_0^{t} ds\;\langle j^s_{\mb z,\mb z+\mb e_i}(s/\ve^2)\big(e_{\mb 0}(0)-\beta^{-1}\big)\\
&=\gamma\;\ve^d \sum_{\mb z, \mb y}\Delta  F(\ve (\mb y+\mb z))H(\ve \mb y)\; 
  \int_0^{t} ds\;\langle \mb p^2_{\mb z}(s/\ve^2) \;\big(e_{\mb 0}(0)-\beta^{-1}\big)\rangle
+\mathcal O (\ve).
\end{split}\end{equation*}
In order to replace in the last expression $p^2_{\mb z}$ with $e_{\mb
  z}$, we use the following Lemma:
\begin{lemma}
For every $G\in L^2(\R^d)$, $\forall d\geq 1$
\begin{equation*}
\lim_{\ve\to 0}\sup_{t\in [0,T]}\left<\left(\int_0^t ds\,
 \ve^{d/2} \sum_{\mb y} G(\ve\mb y)\big[\mb p^2_{\mb y}- e_{\mb y}
\big](s/\ve^2)\right)^2\right> = 0.
\end{equation*}
\end{lemma}
\begin{proof}
Observe that
\begin{equation*}\begin{split}
\mb p_{\mb x}^2 - e_{\mb x}& = \frac 12 \mb p_{\mb x}^2 + \frac 12 \mb
q_{\mb x}\cdot[\alpha\Delta-\nu I]\mb q_{\mb x} 
= \frac 12 L[\mb q_{\mb x}\cdot\mb p_{\mb x}] -\frac 12 \gamma S[\mb
q_{\mb x}\cdot \mb p_{\mb x}]
\end{split}\end{equation*}
Therefore
\begin{equation}\label{eq:Ls}
\begin{split}
&\Big<\Big(\int_0^t ds\,
\ve^{d/2} \sum_{\mb y} G(\ve\mb y)\big[\mb p^2_{\mb y}- e_{\mb y} \big](s/\ve^2)\Big)^2\Big>\\
&\le 
\ve^4 \Big<\Big(\int_0^t ds\,
\frac{\ve^{d/2}}{2} \sum_{\mb y} G(\ve\mb y) \ve^{-2} L [\mb q_{\mb y}\cdot\mb
p_{\mb y}] (s/\ve^2)\Big)^2\Big>\\
& + 
\gamma^2 \Big<\Big(\int_0^t ds\,
\frac{\ve^{d/2}}{2} \sum_{\mb y} G(\ve\mb y)S[\mb q_{\mb y}\cdot\mb p_{\mb y}](s/\ve^2)\Big)^2\Big>.
\end{split}
\end{equation}
In the first term of the rhs of the above inequality, we can perform explicitely
the time integration and we have
\begin{equation*}
  \begin{split}
    \frac{\ve^{d/2}}{2} \sum_{\mb y} G(\ve\mb y) \int_0^t ds\,
    \ve^{-2} L [\mb q_{\mb y}\cdot\mb p_{\mb y}] (s/\ve^2) \\
    = \frac{\ve^{d/2}}{2} \sum_{\mb y} G(\ve\mb y)
    \left( [\mb q_{\mb y}\cdot\mb p_{\mb y}] (t/\ve^2) - [\mb q_{\mb
        y}\cdot\mb p_{\mb y}] (0)\right) + \mc M_\ve(G,t)
  \end{split}
\end{equation*}
where $\mc M_\ve(G,t)$ is a martingale (given by a stochastic
integral) whose quadratic variation is bounded by
\begin{equation*}\begin{split}
  \left[ \mc M_\ve(G,t)^2\right] &\le \ve \frac{\ve^{d}}{4} \sum_{\mb
    y} \left|\nabla^\ve G(\ve\mb y)\right|^2 \int_0^t \langle\mb p_{\mb y}^2 (s/\ve^2)
  \mb q_{\mb y}^2 (s/\ve^2)\rangle ds  \\
& \le \ve\, C\|G\|^2 \beta^{-1}\langle \mb q_{\mb
  0}^2\rangle.
\end{split}\end{equation*}
Since 
$\langle \mb q_{\mb 0}^2\rangle$ is bounded in the pinned case and for
$d\geq 3$, we conclude that the first term on the RHS of \eqref{eq:Ls}
converges to $0$ as $\ve \to 0$. 

Moreover, by the estimate in appendix A
\begin{equation*}\begin{split}
&\left<\left(\int_0^t ds\,
\ve^{d/2} \sum_{\mb y} G(\ve\mb y)S[\mb q_{\mb y}\cdot\mb p_{\mb
  y}](s/\ve^2)\right)^2\right> \le Ct \ve^2
\end{split}
\end{equation*}
which vanishes as $\ve\to 0$.

\end{proof}

Using Cauchy-Schwartz, the quantity
\begin{equation*}\begin{split}
\gamma\;\ve^d \sum_{\mb z, \mb y}\Delta  F(\ve (\mb y+\mb z))H(\ve \mb y)\; 
  \int_0^{t} ds\;\langle \big[\mb p^2_{\mb z}-e_{\mb z}\big](s/\ve^2) \;\big(e_{\mb 0}(0)-\beta^{-1}\big)\rangle\\
=\gamma\;\ve^d \sum_{\mb y', \mb y}\Delta  F(\ve \mb y')H(\ve \mb y)\; 
  \int_0^{t} ds\;\langle \big[\mb p^2_{\mb y'}-e_{\mb y'}\big](s/\ve^2) \;\big(e_{\mb y}(0)-\beta^{-1}\big)\rangle
\end{split}\end{equation*}
is bounded in absolute value by
$$
\gamma \|H\|\langle \big(e_{\mb 0}-\beta^{-1}\big)^2\rangle^{1/2}
\left\langle\left(\int_0^t ds\,
\ve^{d/2} \sum_{\mb y} \Delta F(\ve\mb y)\big[\mb p^2_{\mb y}- e_{\mb y} \big](s/\ve^2)\right)^2\right\rangle^{1/2},
$$
which vanishes as $\ve\to 0$ in view of the previous lemma. Then the contribution of the noise
in the evolution of the energy fluctuations is given by
\begin{equation}\label{noise}\begin{split}
&\gamma \ve^d\sum_{\mb y, \mb z}\sum_{i=1}^d
\nabla^\ve_{\mb e_i}F\big(\ve (\mb y+\mb z)\big)H(\ve \mb y)\frac 1 {\ve}
\int_0^{t} ds\;\langle j^s_{\mb z,\mb z+\mb e_i}(s/\ve^2)\big(e_{\mb 0}(0)-\beta^{-1}\big)\rangle\\
&=\gamma\;\ve^d \sum_{\mb z, \mb y}\Delta  F(\ve (\mb y+\mb z))H(\ve \mb y)\; 
  \int_0^{t} ds\;\langle  e_{\mb z}(s/\ve^2) \;\big(e_{\mb 0}(0)-\beta^{-1}\big)\rangle
+\mathcal O (\ve)\\
&= \sigma_{\ve,t}(\gamma \Delta F, G) +\mathcal O (\ve).
\end{split}\end{equation}

In $d=1$ we have
\begin{equation*}\begin{split}
&\gamma \ve\sum_{\mb y, \mb z}
\nabla^\ve F\big(\ve (y+z)\big)H(y)\frac 1 {\ve}
\int_0^{t} ds\;\langle j^s_{z,z+1}(s/\ve^2)\big(e_{0}(0)-\beta^{-1}\big)\rangle\\
&=\gamma\;\ve\sum_{y}\sum_{z}F''(\ve (y+z))H(\ve y)\; 
  \int_0^{t} ds\;\langle \phi_z(s/\ve^2) \;\big(e_0(0)-\beta^{-1}\big)\rangle+\mc O(\ve),\\
\end{split}\end{equation*}
with $\phi$ defined in \eqref{def:js1}. The last quantity can be written as
\begin{equation*}\begin{split}
&=\gamma\;\ve\sum_{y}\sum_{z} F''(\ve (y+z)) H(\ve y)\; 
\int_0^t ds\;\Big\{\langle p_z^2(s/\ve^2) \;\big( e_0(0)-\beta^{-1}\big)\rangle\\
&+
\frac 1 6\langle \big(p_{z+1}p_{z-1}-2p_{z}p_{z+1}-2p_zp_{z-1}\big)(s/\ve^2)\big(e_0(0)-\beta^{-1}\big)\rangle\Big\}+\mathcal O (\ve),
\end{split}\end{equation*}
where we can replace $p_z^2$ with $e_z$ using same arguments as  in $d\geq 2$. 
We will prove in Section $\ref{sec:BG}$ that 
the second integral vanishes as $\ve\to 0$,  then formula \eqref{noise} holds in all dimensions.

\subsection*{Hamiltonian currents.}

Notice that
\begin{equation}
  \label{eq:3}
   j^a_{0,\mb e_i}= \frac \alpha 2 \left(\mb p_{\mb e_i} - \mb p_{-\mb e_i}\right)\cdot \mb
   q_{\mb e_i} - \frac\alpha 2 \nabla_i \left(\mb p_{0} \cdot \mb q_0 +
    \mb p_{-\mb e_i}\cdot  \mb q_0 \right)
\end{equation}
it is easy to see that the contribution of the gradient terms in
\eqref{main} vanish as $\ve\to 0$, so we have only to consider the
first term.

Let $g_\lambda^j $, $j=1,..,d$ be the solution of the equation
\begin{equation}\label{def:Gl}
\big(\lambda -2\Delta \big) g_\lambda^j (\mb z)=  
\delta(\mb z+ \mb e_j)  - \delta(\mb z- \mb e_j),\qquad d\geq 2,
\end{equation}
with $\lambda>0$ . In $d=1$ the equation reads
\begin{equation}\label{def:Gl1}
 -\frac{1}{3}\Delta\big[ 4g_\lambda(z)+g_\lambda(z+1)+g_\lambda
 (z-1)\big]+ \lambda g_\lambda(z)=
\delta(z+1) - \delta(z-1).
\end{equation}
We observe that $\gl^j$ decays exponentially  fast, and in particular the $\ell^2$-norm is finite.
We by $\hat g_\lambda^j$ its Fourier transform, namely
\begin{equation}\label{hatG}
\hat g _\lambda^j (\mb k) =\frac {2i\sin(2\pi k_j)} {\Phi(\mb k)+\lambda},
\end{equation} 
where $\Phi$ is defined by \eqref{def:phi}.
We define
\begin{equation}\label{def:u}
u^i_\lambda=
\sum_{\mb x} g_\lambda^i  ({\mb x}) \; \mb q_0\cdot \mb p_{\mb{x}},\qquad
i=1,..,d. 
\end{equation}
so that 
\begin{equation*}\begin{split}
\alpha\left(\mb p_{\mb e_i} - \mb p_{-\mb e_i}\right)\cdot \mb
   q_{\mb e_i} = {\alpha} \big(\lambda - S\big)u^i_\lambda
={\alpha} \lambda u^i_\lambda - \frac \alpha{\gamma} L u^i_\lambda
 + \frac{\alpha}{\gamma} A u^i_\lambda.
\end{split}\end{equation*}
Now we show that
the contribution of the first two terms of the above
to \eqref{main} will vanish as $\ve\to 0$, $\lambda\to 0$. We will use the following Lemma.

\begin{lemma}\label{t:ula} For every $F\in L^2(\bb R^d)$, $\forall i=1,..,d$
  \begin{equation}\label{eq:4}
    \Big\langle\Big( \ve^{d/2} \sum_{\bx} F(\ve \bx) \tau_{\bx}
      u^i_\lambda \Big)^2 \Big\rangle \le \beta^{-1} \|F\|^2
\int d\mb k\; 
  \hat\Gamma(\mb k)\;  \left|\hat g_\lambda^i(\mb k)\right|^2.
  \end{equation}
\end{lemma}

\begin{proof}
  By Schwarz inequality 
\begin{equation*}
\Big\langle\Big( \ve^{d/2} \sum_{\bx} F(\ve \bx) \tau_{\bx}
      u^i_\lambda \Big)^2 \Big\rangle 
\leq\| F\|^2  \sum_\bx \left< u^i_\lambda \tau_{\bx}
          u^i_\lambda\right>,
\end{equation*}
where
  \begin{equation*}
    \begin{split}
     \sum_\bx \left< u^i_\lambda \tau_{\bx}
          u^i_\lambda\right> &= \beta^{-1} \sum_\bx
        \Gamma(\bx) \sum_i \sum_\bz g_\lambda^i({\mb z})
         g_\lambda^i({\mb z+ \mb x}) \\
 &=\beta^{-1}  \int d\mb k\; 
       \hat\Gamma(\mb k)\;  \left|\hat g_\lambda^i(\mb k)\right|^2
    \end{split}
  \end{equation*}
\end{proof}

Notice that if $d\ge 3$ or $\nu>0$ the quantity
$\int d\mb k\; 
       \hat\Gamma(\mb k)\;  \left|\hat g_\lambda^i(\mb k)\right|^2$ is bounded. 

In order to to prove that 
$L u^i_\lambda$ does not contribute in the limit we have just to show that
\begin{equation*}
\big\langle\Big(\ve^{-1}\int_0^t \ve^{d/2} \sum_{\bx} F(\ve \bx) \tau_{\bx}
     L u^i_\lambda (s/\ve^2) \Big)^2\big\rangle \to 0,\qquad \ve\to 0.
\end{equation*}
We can perform the time integration and we deduce
\begin{equation*}\begin{split}
&\ve^{-1}\int_0^t \ve^{d/2} \sum_{\bx} F(\ve \bx) \tau_{\bx}
     L u^i_\lambda (s/\ve^2) \\
&= \ve\, \ve^{d/2} \sum_{\bx} F(\ve \bx) \tau_{\bx}\big[u^i_\lambda(t/\ve^2)
-u^i_\lambda(0)  \big] + \ve \mathcal M_\ve (F,t),
\end{split}\end{equation*}
 where $\mathcal M_\ve(F,t) $ is a martingale given by a stochastic integral. By Lemma \ref{t:ula} the first term on the rhs of the previous equality is of order $\ve$. Moreover, $\mathcal M_\ve(F,t) $ has bounded quadratic variation, therefore the contribution of $\ve \mathcal M_\ve (F,t)$ vanishes as  $\ve\to 0$.

\begin{lemma}\label{t:ula2}
For every $F\in L^2(\bb R^d)$, $\forall i=1,..,d$
\begin{equation*}\begin{split}
\lim_{\lambda\to 0}\lim_{\ve\to 0}\big\langle \Big(\frac {\lambda} {\ve}\int_0^t ds\,\ve^{d/2}\sum_{\mb z}
F\big(\ve \mb y)
\tau_{\mb y} u^i_\lambda(s/\ve^2) \Big)^{2}\big\rangle =0
\end{split}\end{equation*}
\end{lemma}

\begin{proof}

The following inequality holds:
\begin{equation*}\begin{split}
&\sup_{t\in[0,T]}\Big\langle \Big(\frac {\lambda} {\ve}\int_0^t ds\,\ve^{d/2}\sum_{\mb z}
F\big(\ve \mb y)
\tau_{\mb y} u^i_\lambda(s/\ve^2) \Big)^{2}\Big\rangle\\
&\leq C T\sup_{f}\left\{ 2\frac {\lambda}{ \ve}\,
\ve^{d/2}\sum_{\mb y}F(\ve\mb y)
\langle f\tau_{\mb y} u^i_\lambda \rangle - \frac 1 {\ve^2}\langle f (-Sf)\rangle\right\}.
\end{split}\end{equation*}

Observe that 
\begin{equation*}
u^i_\lambda=\frac 1 {d-1}\sum_{\mb x}\sum_{k,j}G_\lambda(\mb x)
q^j_{\mb 0}\big[X_{\mb x-\mb e_i}^{k,j}p^k_{\mb x-\mb e_i}+X_{\mb x+\mb e_i}^{k,j}p^k_{\mb x+\mb e_i}\big]
\end{equation*}
for $d\geq 2$,
where $G_\lambda(\mb x)$ solves $(\lambda-2\Delta)G_\lambda(\mb x)=\delta(\mb x)$. A similar formula holds for $d=1$.
 Then 
\begin{equation*}\begin{split}
&\frac {\lambda}{\ve}\,\ve^{d/2}\sum_{\mb y}F(\ve\mb y)
\langle f\tau_{\mb y} u^i_\lambda \rangle\\
& =
\frac 1{d-1}\frac {\lambda}{\ve}\,\ve^{d/2}\sum_{\mb y}F(\ve\mb y)
 \sum_{\mb x}\sum_{k,j}G_\lambda(\mb x-\mb y)\big\{
\big\langle \big(X_{\mb x-\mb e_i}^{k,j}f\big)\, q^j_{\mb y}p^k_{\mb x-\mb e_i}\big\rangle\\
& \qquad +
\big\langle \big(X_{\mb x-\mb e_i}^{k,j}f \big)\,q^j_{\mb y}p^k_{\mb x-\mb e_i}\big\rangle
\big\}
\end{split}\end{equation*}
which is bounded in absolute value by 
\begin{equation*}\begin{split}
C\left(\lambda^2\beta^{-1}\sum_{\mb x}\sum_j
\big\langle\Big(\ve^{d/2} \sum_{\mb y}F(\ve\mb y)
G_\lambda(\mb x-\mb y)q_{\mb y}^j  \Big)^2\big\rangle
\right)^{1/2}\\
\times \Big(\ve^{-2}\langle f(-S\,f)\rangle\Big)^{1/2}.\\
\end{split}\end{equation*}
Observe that
\begin{equation*}\begin{split}
&\lambda^2\sum_{\mb x}\sum_j
\big\langle\Big(\ve^{d/2} \sum_{\mb y}
F(\ve\mb y)
G_\lambda(\mb x-\mb y)q_{\mb y}^j  \Big)^2\big\rangle\\
&=\lambda^2\ve^d\sum_{\mb y, \mb z}F(\ve\mb y)
F\big(\ve(\mb y+\mb z)\big)\sum_{\mb x}G_\lambda(\mb x)G_\lambda(\mb x-\mb z)\langle\mb q_{\mb z}\cdot \mb q_{\mb 0} \rangle\\
&\leq \lambda^2 \|F\|^2\sum_{\mb x, \mb z}G_\lambda(\mb x)G_\lambda(\mb x-\mb z)\Gamma(\mb z),
\end{split}\end{equation*}
where in he last step we used Schwarz inequality. In the Fourier space
\begin{equation*}
\lambda^2 \sum_{\mb x, \mb z}G_\lambda(\mb x)G_\lambda(\mb x-\mb z)\Gamma(\mb z)=
\lambda^2 \int_{\bb T^d} d\mb k \frac {\hat\Gamma (\mb k)}{\big(\lambda+\Phi(\mb k)\big)^2},
\end{equation*}
which vanishes as $\lambda\to 0$ for $\nu>0$ or $d\geq 3$.
Therefore 
\begin{equation*}\begin{split}
&\sup_{f}\left\{ 2\frac {\lambda}{ \ve}\,
\ve^{d/2}\sum_{\mb y}F(\ve\mb y)
\langle f\tau_{\mb y}  u^i_\lambda \rangle - \frac 1 {\ve^2}\langle f (-Sf)\rangle\right\}\\
&\leq  C_0\big\| F\big\|^2 \lambda^2 \sum_{\mb x, \mb z}G_\lambda(\mb x)G_\lambda(\mb x-\mb z)\Gamma(\mb z)
\end{split}\end{equation*}
which vanishes for $\lambda\to 0$. 

\end{proof}

Observe that 
\begin{equation*}\begin{split}
\frac{\alpha}{\gamma} A u^i_\lambda=& \frac{\alpha}{\gamma}
\sum_{\mb x} g_\lambda^i ({\mb x}) 
\Big\{\mb p_{\mb 0}\cdot \mb p_{\mb x} +
\mb q_{\mb 0} \cdot (\alpha\Delta-\nu I)\mb q_{\mb x}\Big\}.
\end{split}\end{equation*}
So we have to look at
\begin{equation*}\begin{split}
&\frac{\alpha}{\gamma}\ve^{d-1} \sum_i\sum_{\mb y,\mb z}
\nabla^{\ve}_{\mb e_i}F\big(\ve (\mb 
y + \mb z)\big) H(\ve \by) \tau_\bz A u^i_\lambda\\ 
&= \frac{\alpha}{\gamma} \ve^{d-1}\sum_i \sum_{\mb y,\mb z}
\nabla^{\ve}_{\mb e_i}F\big(\ve (\mb y + \mb z)\big) H(\ve \by) 
\sum_{\mb x} g_\lambda^i({\mb x})
\Big[\mb p_{\mb z}\cdot \mb p_{\mb x+ \mb z} +
\mb q_{\mb z}\cdot (\alpha\Delta-\nu I)\mb q_{\mb x+ \mb z}\Big] \\
&= \frac{\alpha}{\gamma} \ve^{d-1} \sum_i \sum_{\mb y,\mb z}
\nabla^{\ve}_{\mb e_i}F\big(\ve (\mb y + \mb z)\big) H(\ve \by) 
\sum_{\mb x} g_\lambda^i({\mb x -\mb z})
\Big[\mb p_{\mb z}\cdot \mb p_{\mb x} +
\mb q_{\mb z}\cdot (\alpha\Delta-\nu I)\mb q_{\mb x}\Big] 
\end{split}
\end{equation*}
We are left to study:
\begin{equation*}\begin{split}
K_2 =& \sum_i \frac{\alpha}{\gamma}\ve^d\sum_{\mb y,\mb z}H(\ve\mb
y)\nabla^\ve_i F\big(\ve(\mb y+\mb z)\big)\sum_{\mb x}g^i_\lambda(\mb
x)\\ 
&\times\frac 1 {\ve}\int_0^t ds\,\langle\tau_{\mb z}[\mb p_{\mb 0}\cdot\mb p_{\mb x}+\mb q_{\mb 0}\cdot\big(\alpha\Delta-\nu I  \big)\mb q_{\mb x}] (s/\ve^2)\big(e_{\mb 0}(0)-\beta^{-1}\big)\rangle.
\end{split}\end{equation*}

Remark that 
\begin{equation*}
  \begin{split}
    A\,\frac 1 2 \big[\mb p_{\mb 0}\cdot \mb q_{\mb x}+ \mb q_{\mb
      0}\cdot \mb p_{\mb x}\big]=&\frac 1 2 \big[\mb q_{\mb
      0}\cdot\big(\alpha\Delta-\nu I \big)\mb q_{\mb x}+
    (\alpha\Delta-\nu I  \big)\mb q_{\mb 0}\cdot \mb q_{\mb x}\big]
    +\mb p_{\mb 0}\cdot\mb p_{\mb x},
  \end{split}
\end{equation*}
that implies
\begin{equation*}
  \begin{split}
    &\mb p_{\mb 0}\cdot\mb p_{\mb x}+\mb q_{\mb 0}
    \cdot\big(\alpha\Delta-\nu I  \big)\mb q_{\mb x} \\
    &= \frac 1 2 \big[\mb q_{\mb
      0}\cdot\big(\alpha\Delta-\nu I \big)\mb q_{\mb x} -
    (\alpha\Delta-\nu I  \big)\mb q_{\mb 0}\cdot \mb q_{\mb x}\big] +
   \frac 1 2 A \big[\mb p_{\mb 0}\cdot \mb q_{\mb x}+ \mb q_{\mb
      0}\cdot \mb p_{\mb x}\big]  \\
    &=  \frac 1 2 \sum_j \nabla_{\mb e_j} \left[ \mb q_{\mb x}\cdot \mb
      q_{-\mb e_j} - \mb q_{\mb 0}\cdot \mb q_{\mb x-\mb e_j} \right]
    + \frac 12 (L-\gamma S) \big[\mb p_{\mb 0}\cdot \mb q_{\mb x}+ \mb q_{\mb
      0}\cdot \mb p_{\mb x}\big] 
  \end{split}
\end{equation*}
The last term on the right hand side gives a negligible contribution
as $\ve \to 0$.
The first term is a gradient and by summation by part gives 
\begin{equation}
  \label{eq:6}
  \begin{split}
    \frac{\alpha}{\gamma}\ve^d\sum_{i,j} \sum_{\mb y,\mb z} & H(\ve\mb
    y)\nabla^{\ve,*}_j \nabla^\ve_i F\big(\ve(\mb y+\mb z)\big)
    \sum_{\mb x}g^i_\lambda(\mb x)\\
    & \int_0^t ds\,\langle\tau_{\mb z}[\mb q_{\mb x}\cdot \mb q_{-\mb
      e_j} - \mb q_{\mb 0}\cdot \mb q_{\mb x-\mb e_j}]
    (s/\ve^2)\big(e_{\mb 0}(0)-\beta^{-1}\big)\rangle \\
    = \frac{\alpha}{\gamma}\ve^d\sum_{i,j} \sum_{\mb y,\mb z} & H(\ve\mb
    y)\nabla^{\ve,*}_j \nabla^\ve_i F\big(\ve\mb z\big)
    \sum_{\mb x}g^i_\lambda(\mb x)\\
    & \int_0^t ds\,\langle\tau_{\mb z}[\mb q_{\mb x}\cdot \mb q_{-\mb
      e_j} - \mb q_{\mb 0}\cdot \mb q_{\mb x-\mb e_j}]
    (s/\ve^2)\big(e_{\mb 0}(\mb y)-\beta^{-1}\big)\rangle
  \end{split}
\end{equation}

Observe that 
\begin{equation}\begin{split}
&\sum_{\mb z} G(\ve\mb z)\sum_{\mb x}g^j_\lambda(\mb x)
\tau_{\mb z}[\mb q_{\mb x}\cdot \mb q_{-\mb
      e_i} - \mb q_{\mb 0}\cdot \mb q_{\mb x-\mb e_i}]\\
&= \sum_\bz G(\ve\mb z)\sum_{\mb x}\big[g^j_\lambda(\mb x-\mb
e_i)-g^j_\lambda(\mb x+\mb e_i)\big] 
\tau_{\mb z}[\mb q_{\mb x}\cdot \mb q_{\mb 0}](1+\mc O(\ve)).
\end{split}\end{equation}
We denote by $c^{i,j}_\lambda(\mb x):=g^j_\lambda(\mb x-\mb e_i)-g^j_\lambda(\mb x+\mb e_i)$. 
Let $f(\mb z)$ be the function satisfiying
\begin{equation*}
(\alpha\Delta-\nu I)f(\mb z)=\delta(\mb z).
\end{equation*}
By direct computation
\begin{equation}\begin{split}
\sum_{\mb x}c^{i,j}_\lambda(\mb x)\mb q_{\mb x}\cdot \mb q_{\mb 0}= &
\big(L-\gamma S\big)\big[\mb q_{\mb 0}\cdot\sum_{\mb z}\sum_{\mb x}c^{i,j}_\lambda(\mb x)f(\mb z-\mb x) \mb p_{\mb z}\big]\\
& -\mb p_{\mb 0}\cdot\sum_{\mb z}\sum_{\mb x}c^{i,j}_\lambda(\mb x)f(\mb z-\mb x) \mb p_{\mb z}.
\end{split}\end{equation}
The function $\tilde \kappa^{i,j}_\lambda(\mb z):=-\sum_{\mb x}c^{i,j}_\lambda(\mb x)f(\mb z-\mb x)$
decay exponentially fast, since the Fourier transfom is given by
\begin{equation*}
\frac{4\sin(2\pi k_i)\sin(2\pi k_j)}{\omega^2(\mb k)}\frac 1 
{\Phi(\mb k)+\lambda}.
\end{equation*}
Therefore one can show that the the term $(L-\gamma S)(\cdot)$ gives a contribution of order $\ve$ and \eqref{eq:6} is equal to
\begin{equation}\label{cond2}\begin{split}
& \frac 1 {8\pi^2\,\gamma}\ve^d \sum_{\mb y, \mb z}\sum_{j,\ell} H(\ve \mb y)
\nabla_j^{\ve,*}\nabla_\ell^\ve F\big(\ve(\mb y+\mb z)\big)
\int_0^t ds\;\sum_{\mb x}\kappa_\lambda^{j,\ell}(\mb x) \\
&\;\;\times
\langle\big[\mb p_{\mb z}\cdot\mb p_{\mb x+\mb z}\big](s/\ve^2)\big(e_{\mb 0}(0)-\beta^{-1}\big) \rangle\;+\mathcal{O}(\ve),
\end{split}\end{equation}
where 
$$
\kappa^{j,\ell}_\lambda(\mb x)=\int_{\bb T^d} d\xi\;
\frac{\partial_j\omega^2(\xi)\partial_\ell
  \omega^2(\xi)}{\omega^2(\xi)} \frac 1 
{\Phi(\xi)+\lambda}\,e^{2i\pi \mb x\cdot \xi },\qquad \ell, i=1,..,d.
$$

Consequently 
the quantity \eqref{cond2} is asymptotic  to
\begin{equation}
\label{cond3}
\begin{split}
& \frac 1 {8\pi^2\,\gamma}\ve^d \sum_{\mb y, \mb z}\sum_{j,\ell} H(\ve \mb y)
\nabla_j^{\ve,*}\nabla_\ell^\ve F\big(\ve(\mb y+\mb z)\big)\\
&\;\;\times\int_0^t ds\; \kappa_\lambda^{j,\ell}(\mb 0) 
\langle e_{\mb z}(s/\ve^2)\big(e_{\mb 0}(0)-\beta^{-1}\big) \rangle\\
& +\frac 1 {8\pi^2\,\gamma}\ve^d \sum_{\mb y, \mb z}\sum_{j,\ell} H(\ve \mb y)
\nabla_j^{\ve,*}\nabla_\ell^\ve F\big(\ve(\mb y+\mb z)\big)\\
&\;\;\times\int_0^t ds\;\sum_{\mb x\neq \mb 0} \kappa_\lambda^{j,\ell}(\mb x) 
\langle\big[ \mb p_{\mb z}\cdot\mb p_{\mb x+\mb z}\big](s/\ve^2)\big(e_{\mb 0}(0)-\beta^{-1}\big) \rangle.
\end{split}\end{equation}
We observe that $\kappa_\lambda^{j,\ell}(\mb 0)=0$ if $j\neq \ell$ and $\kappa_\lambda^{j,j}(\mb 0)=
\kappa_\lambda^{1,1}(\mb 0)$ $\forall j=1,..,d$. We set
\begin{equation}\label{def:tk}
\kappa_\lambda:=\kappa_\lambda^{1,1}(\mb 0)=\int_{\bb T^d}d\xi \,
\Big(\frac{\partial_1 \omega^2(\xi)}{\omega(\xi)}\Big)^2\frac 1
{\Phi(\xi)+\lambda}. 
\end{equation}
Then \eqref{cond3} is equal to
\begin{equation}
\begin{split}
& \frac 1 {8\pi^2\,\gamma}\ve^d \sum_{\mb y, \mb y'}H(\ve \mb y)
\kappa_\lambda\Delta F(\ve\mb y')\int_0^t ds\;  
\langle e_{\mb y'}(s/\ve^2)\big(e_{\mb y}(0)-\beta^{-1}\big) \rangle\\
& +\frac 1 {8\pi^2\,\gamma}\ve^d \sum_{\mb y, \mb z}\sum_{j,\ell} H(\ve \mb y)
\nabla_j^{\ve,*}\nabla_\ell^\ve F\big(\ve(\mb y+\mb z)\big)\\
&\;\;\times\int_0^t ds\;\sum_{\mb x\neq \mb 0} \kappa_\lambda^{j,\ell}(\mb x) 
\langle\big[ \mb p_{\mb z}\cdot\mb p_{\mb x+\mb z}\big](s/\ve^2)\big(e_{\mb 0}(0)-\beta^{-1}\big) \rangle.
\end{split}
\end{equation}
We will prove in the next section that the second term vanishes as $\ve\to 0$, $\lambda\to 0$.

\section{Boltzmann-Gibbs principle}\label{sec:BG}

By using as above the Schwarz inequality, all we need to prove is that
\begin{equation}\label{eq:BG0} 
\sum_{\mb x \neq \mb 0}|\kappa^{i,\ell}_\lambda(\bx)
|{\left<\left(\int_0^t ds\,\ve^{d/2}
\sum_\by  F(\ve \by) \mb p_{\by} \cdot \mb p_{\by+\bx}(s/\ve^2)  \right)^2\right>}^{1/2}.
\end{equation}
is negligeable as $\ve\to 0$.

We denote the cube in  $\Z^d$ of size $2\ell+1$ by
$\Lambda_\ell:=\{\mb z\in\Z^d:\, |z^j|\leq \ell, j=1,\dots,d\}$, and
for every $x\in\Z$ we define 
  $\Psi_{\ell,x}:=\frac 1 {|\Lambda_\ell|}\sum_{\mb z\in \Lambda_\ell}
  \mb p_\bz \cdot\mb p_{\bz+\bx}$.
We observe that in \eqref{eq:BG0} we can replace $\bp_{\by}\cdot
\bp_{\by+\bx}$ with $\tau_\by \Psi_{\ell,\bx}$, with a difference   $\sim \ve |\Lambda_\ell|\, \|\kappa_\lambda\|_{\ell^1}$.
We denote by $\langle\cdot\rangle_{\Lambda_K}=\langle
\cdot\rangle_{\Lambda_K,\, \mathcal T_K,  
\mathcal P_K}$ the micro-canonical expectation in the box $\Lambda_K$,
with $\mathcal T_K=\sum_{\bx\in \Lambda_K} \bp_\bx^2$ and $\mathcal
P_K=\sum_{\bx\in\Lambda_K} \bp_\bx$. 
We define
$\widetilde\Psi_{\ell, \bx}:= \Psi_{\ell,\bx}-\langle \Psi_{\ell,\bx}
\rangle_{\Lambda_{\ell+|\bx|}}$.  
Then
\begin{equation}\label{eq:var}
\begin{split}
&\left<\left( \int_0^t ds\, \ve^{d/2}
\sum_\by  F(\ve \by) \tau_\by\widetilde \Psi_{\ell,\bx}(s/\ve^2) 
 \right)^2\right>\\
&
\leq C t \sup_{f}\left\{
\ve^{d/2} \sum_\by  F(\ve \by)\left< f\,\tau_y\widetilde\Psi_{\ell,\bx}\right>
-  \ve^{-2}\left< f\,(-S f) \right>\right\}.
\end{split}
\end{equation}
Introduce $S_{\Lambda_K}= \displaystyle\frac  1{4(d-1)} \sum_{\bx, \bz \in \Lambda_K \atop \|\bx - \bz\|=1}\sum_{i,j}
   \left( X^{i,j}_{\bx, \bz}\right)^2$. 
By the spectral gap of $S_K$, 
there exists $\tilde \Psi_{\ell,\bx}= S_{\Lambda_{\ell+|\bx|}}U_{\ell,
  \bx}$, $\forall \bx$, $\forall \ell$. Moreover, since the spectral
gap of $S_{\Lambda_K} $ is bounded below by $CK^{-2}$, we have that 
$\langle U_{\ell,\bx} \widetilde\Psi_{\ell,\bx}\rangle^2 \leq
C(\ell+|\bx|)^2\langle \widetilde\Psi_{\ell,\bx}^2\rangle \leq
C(\ell+|\bx|)^2 \beta^{-1} \ell^{-d}$. 

Then since
\begin{equation*}\begin{split}
\langle f\,\tau_y\tilde\Psi_{\ell,x}\rangle \leq \langle U_{\ell, x}
\Psi_{\ell,x}\rangle^{1/2} 
\langle \tau_\by f\,(-S_{\Lambda_{\ell+|\bx|}} \tau_\by f) \rangle \big)^{1/2}.
\end{split}\end{equation*}
 we can bound the right hand side of \eqref{eq:var} by
\begin{equation*}
  \begin{split}
    C t \sum_\by \sup_{f}\left\{
\ve^{d/2}  F(\ve \by)\left< f\,\tau_y\widetilde\Psi_{\ell,\bx}\right>
-  \ve^{-2}|\Lambda_{\ell+|\bx|}|^{-1}\left< \tau_\by f\,(-S_{\Lambda_{\ell+|\bx|}} \tau_\by f)
\right>\right\}\\
\le 
 C t \sum_\by \sup_{f}\Big\{
\ve^{d/2}  F(\ve \by) C(\ell+|\bx|)\langle
\widetilde\Psi_{\ell,\bx}^2\rangle^{1/2} \langle
 \tau_\by f\,(-S_{\Lambda_{\ell+|\bx|}} \tau_\by f) \rangle \big)^{1/2} \\
-  \ve^{-2}|\Lambda_{\ell+|\bx|}|^{-1}\left< \tau_\by f\,(-S_{\Lambda_{\ell+|\bx|}} \tau_\by f)
\right>\Big\}\\
\le
 C' t \ve^{d+2} \sum_\by  F(\ve \by)^2 (\ell+|\bx|)^{d+2} \langle
\widetilde\Psi_{\ell,\bx}^2\rangle.
  \end{split}
\end{equation*}
Since $\sum_\bx \kappa^{i,\ell}_\lambda(\bx) (\ell+|\bx|)^{(d+2)/2} <\infty$ 
 we
conclude that 
\begin{equation*}
 \lim_{\ve\to 0} \sum_{\mb x \neq \mb 0}|\kappa^{i,\ell}_\lambda(\bx)
|{\left<\left(\int_0^t ds\,\ve^{d/2}
\sum_\by  F(\ve \by) \tau_\by \widetilde\Psi_{\ell,\bx}
(s/\ve^2)  \right)^2\right>}^{1/2} =0.
\end{equation*}

Setting $\bar\Psi_{\ell,\bx}:=\langle \Psi_{\ell,\bx} \rangle_{\Lambda_{\ell+|\bx|}}$,
now we have just to show that 
\begin{equation*}\begin{split}
\lim_{\ve\to 0}\sum_{\bx\neq
  0}|\kappa_\lambda(\bx)|{\left<\left(\int_0^t ds\,
\ve^{d/2}
\sum_\by F(\ve \by) \tau_\by\bar\Psi_{\ell,\bx}(s/\ve^2)  \right)^2\right>}^{1/2}=0.
\end{split}\end{equation*}

The previous expression is bounded by
\begin{equation*}
t\,\sum_{\bx\neq 0}|\kappa_\lambda(\bx)|{\left<\left(
\ve^{d/2} \sum_\by  F(\ve \by)\tau_\by \bar\Psi_{\ell,\bx}
\right)^2\right>}^{1/2}.
\end{equation*}
We observe that
\begin{equation*}\begin{split}
&\left<\Big(
\ve^{d/2}\sum_\by  F(\ve \by)\tau_\by \bar\Psi_{\ell,\bx}
\Big)^2\right>\\
&  \leq \ve^d \sum_{\by,\,\by'} F(\ve \by)F(\ve \by')
\langle \tau_{\by-\by'} \bar\Psi_{\ell,\bx}\,\bar\Psi_{\ell,\bx}\rangle\\
&\leq \frac{\ve^d}2 \sum_{\by,\,\by'}\big( F(\ve \by)^2+F(\ve \by')^2\big)
\langle \tau_{\by- \by'} \bar\Psi_{\ell,\bx}\,\bar\Psi_{\ell,\bx}\rangle\\
&\leq \| F\|^2_{L^2(\R)}\sum_{|\bz|\leq 2(\ell+|\bx|)}
\langle \tau_{z} \bar\Psi_{\ell,x}\,\bar\Psi_{\ell,x}\rangle\\
& \leq C \| F\|^2_{L^2(\R)}(\ell+|\bx|)^d\langle \bar\Psi_{\ell,\bx}^2  \rangle.
\end{split}\end{equation*}
By the properties of the microcanonical measure it holds
$\langle\bar\Psi_{\ell,x}^2\rangle\leq C_0 \beta^{-2}(\ell+|\bx|)^{-2d}$,
and therefore
\begin{equation*}
\begin{split}
&
t\,\sum_{x,\,x\neq 0}|\kappa_\lambda(x)|{\big\langle\Big(
\ve^{d/2}\sum_y  F'(\ve y)\tau_y \bar\Psi_{\ell,x}
\Big)^2\big\rangle}^{1/2}\\
&\leq C_1 T \beta^{-1}\| F'\|_{L^2(\R)} \|\kappa_\lambda\|_{\ell^1}\frac 1 { \ell^{d/2}},
\end{split}
\end{equation*}
which vanishes as $\ell\to\infty$.

\section{Appendix A}
\label{sec:appendixA}

Let $\phi(\bx,\by)$ a local function on $\Z^d\times \Z^d$ and define 
\begin{equation*}
  \Phi = \sum_{\bx,\by} \phi(\bx,\by) \bp_\bx \cdot \bq_\by
\end{equation*}
Consider $F\in L^2(R^d)$.
We want to find a good upper bound for the variance:
\begin{equation}\label{varA}
  \left<\left(\int_0^t ds \; \ve^{d/2} \sum_{\by} F(\ve \by) \tau_y S\Phi
      (s/\ve^2)\right)^2\right> 
\end{equation}
By lemma (2.4) in \cite{klo}, pag. 48, we have
\begin{equation*}
  \begin{split}
    &\left<\left(\sup_{0\le s\le T} \int_0^t ds \; \ve^{d/2} \sum_{\by}
        F(\ve \by) \tau_y S\Phi (s/\ve^2)\right)^2\right> \\
    &\le 24 T
    \ve^2 \left< \left(\ve^{d/2} \sum_{\by} F(\ve \by) \tau_y
        \Phi\right) (-S) \left(\ve^{d/2} \sum_{\by} F(\ve \by) \tau_y
        \Phi) \right)\right>\\
    &= \frac{48 T}{\beta} \ve^{d+2} \sum_{\by, \by'} F(\ve \by)  F(\ve
    \by')  \Xi(\by-\by') 
  \end{split}
\end{equation*}
where 
\begin{equation}
  \label{eq:1}
  \Xi(\by) = \sum_{\bx, \bx', \bz} \phi(\bx, \bx') (\Delta_1
  \phi)(\bx + \by, \bz) \Gamma(\by + \bx' - \bz)
\end{equation}
where $\Delta_1 \phi$ indicates the discrete laplacian on the first
variable of $\phi$.

Then by Schwarz inequality we can bound \eqref{varA} by 
\begin{equation}
  \label{eq:2}
  \|F\|^2 \ve^2 \sum_{\by} |\Xi(\by) | .
\end{equation}

\section{formulaire}
\label{sec:formulaire}

Some formulas we use:

\begin{equation}
  \label{eq:11}
  S \bp_\bx  = 2 \Delta \bp_\bx  \qquad d\ge 2
\end{equation}
\begin{equation}
  \label{eq:12}
  S p_x = \frac 16 \Delta \left( 4 p_x + p_{x+1} + p_{x-1} \right)
  \qquad d=1
\end{equation}
\begin{equation}
  \label{eq:14}
   S e_\bx = S \bp_\bx^2/2  =  \Delta \bp_\bx  \qquad d\ge 2
\end{equation}
\begin{equation}
  \label{eq:13}
   S e_x =  S p_x^2/2 = \frac 16 \Delta \left( 4 p_x^2 + p_{x+1}^2 +
     p_{x-1}^2 \right)  \qquad d=1
\end{equation}

\end{document}